\pdfoutput=1
\RequirePackage{ifpdf}
\ifpdf % We are running pdfTeX in pdf mode
\documentclass[pdftex]{sigma}
\else
\documentclass{sigma}
\fi

\numberwithin{equation}{section}

\newtheorem{Theorem}{Theorem}[section]
\newtheorem{Lemma}[Theorem]{Lemma}
 { \theoremstyle{definition}
\newtheorem{Example}[Theorem]{Example} }

\renewcommand{\mod}{\mathbin{\mathrm{mod}}}
\DeclareMathOperator{\rem}{rem}
\DeclareMathOperator{\diag}{diag}

\begin{document}
\allowdisplaybreaks

\newcommand{\arXivNumber}{1811.09096}

\renewcommand{\PaperNumber}{045}

\FirstPageHeading

\ShortArticleName{Lax Representations for Separable Systems from Benenti Class}

\ArticleName{Lax Representations for Separable Systems\\ from Benenti Class}

\Author{Maciej B{\L}ASZAK~$^\dag$ and Ziemowit DOMA\'NSKI~$^\ddag$}

\AuthorNameForHeading{M.~B{\l}aszak and Z.~Doma\'nski}

\Address{$^\dag$~Faculty of Physics, Division of Mathematical Physics, A.~Mickiewicz University,\\
\hphantom{$^\dag$}~Uniwersytetu Pozna\'{n}skiego 2, 61-614 Pozna\'{n}, Poland}
\EmailD{\href{mailto:blaszakm@amu.edu.pl}{blaszakm@amu.edu.pl}}

\Address{$^\ddag$~Institute of Mathematics, Pozna{\'n} University of Technology, \\
\hphantom{$^\ddag$}~Piotrowo 3A, 60-965 Pozna{\'n}, Poland}
\EmailD{\href{mailto:ziemowit.domanski@put.poznan.pl}{ziemowit.domanski@put.poznan.pl}}

\ArticleDates{Received December 13, 2018, in final form June 07, 2019; Published online June 18, 2019}

\Abstract{In this paper we construct Lax pairs for St\"ackel systems with separation curves from so-called Benenti class. For each system of considered family we present an infinite family of Lax representations, parameterized by smooth functions of spectral parameter.}


\Classification{70H06; 37J35}

\section{Introduction}\label{sec1}
Classical St\"{a}ckel systems belong to important class of integrable and separable Hamiltonian ODE's. The constants of motion of these systems are quadratic in momenta and describe many physical systems of classical mechanics. The St\"{a}ckel systems are defined by separation relations, i.e., relations involving canonical variables $(\lambda_i ,\mu_i)_{i=1,\dots,n}$ in which the Hamilton--Jacobi equations separate to a system of decoupled ordinary differential equations. The separation relations of classical St\"{a}ckel system are represented by $n$ algebraic equations of the form
\begin{gather}
\sigma _{i}(\lambda _{i})+\sum_{k=1}^{n}H_{k}S_{ik}(\lambda _{i})=\frac{1}{2}%
f_{i}(\lambda _{i})\mu _{i}^{2},\qquad i=1,2,\dots ,n, \label{eq1}
\end{gather}
where $n$ is the number of degrees of freedom of the system, i.e., $2n$ is
the dimension of the corresponding phase space on which the system is
defined, $H_{1},H_{2},\dots ,H_{n}$ are $n$ Hamiltonians, $f$, $\sigma $
and $S$ are arbitrary smooth functions. Solving the system~(\ref{eq1}), under
assumption that $\det S\neq 0$, with respect to all $H_{i}$ we get $n$
Hamiltonians expressed in variables $(\lambda_i ,\mu_i )_{i=1,\dots,n}$, which
from construction will be in involution, i.e., their Poisson brackets vanish
$\{H_{i},H_{j}\}=0$, and which Hamilton--Jacobi equations separate. In other
words the system~(\ref{eq1}) describes a~Liouville-integrable and separable
Hamiltonian system.

The special, but particularly important class of St\"{a}ckel systems is the
Benenti class \cite{Benenti1992,Benenti1997,Benenti2005}. This class is
described by the following separation relations
\begin{gather}
\sigma _{i}(\lambda _{i})+H_{1}\lambda _{i}^{n-1}+H_{2}\lambda
_{i}^{n-2}+\dotsb +H_{n}=\frac{1}{2}f_i(\lambda _{i})\mu _{i}^{2},\qquad
i=1,2,\dots ,n. \label{eq2}
\end{gather}

There is an extended literature on systems from class (\ref{eq2}), however
less can be found on their Lax representation. A Lax representation of a
Liouville integrable Hamiltonian system is a set of matrices~$L$, $U_{k}$ ($k=1,2,\dots ,n$) which satisfy the system of Lax equations
\begin{gather}
\frac{{\rm d}L}{{\rm d}t_{k}} = [U_{k},L],\qquad k=1,2,\dots,n, \label{eq3}
\end{gather}%
under the assumption that the time evolution with respect to $t_{k}$ is
governed by the Hamiltonian vector field $X_{H_{k}}$.

In the paper we find an infinite family of Lax representations for an
arbitrary St\"{a}ckel system from the Benenti class generated by the
following canonical form of separation curves
\begin{gather}
\sigma (\lambda )+\sum_{k=1}^{n}H_{k}\lambda ^{n-k}=\frac{1}{2}f(\lambda
)\mu ^{2}, \label{eq4}
\end{gather}%
where the spectral parameter $\lambda$ is real and the functions $f(\lambda)$
and $\sigma(\lambda)$ are smooth and real valued. Separation relations (\ref%
{eq2}) are reconstructed by $n$ copies of (\ref{eq4}) with $(\lambda _{i},\mu
_{i})_{i=1,\dots ,n}$. In literature, the reader can find Lax
representation for systems related to various subcases of separation curves
from the family~(\ref{eq4}). First constructions of Lax representation for
separable systems with separation curves of particular hyperelliptic form
was constructed by Mumford~\cite{Mumford1984}. For~$\sigma (\lambda )$ in polynomial
form and $f(\lambda )=\lambda ^{2r}$, $r\in \mathbb{N}$, Lax equations were
constructed expli\-citly in~\cite{Vanhaecke2001} and analyzed in particular coordinate
frames. In~\cite{Eilbeck1994}, using different technique, Lax representation was
constructed for the subclass of separation curves (\ref{eq4}) with $f(\lambda
) $ being polynomials of order $n+1$, $n$ and $n-1$, respectively, with
distinct roots. Yet another subcases of family (\ref{eq4}) together with the
construction of Lax representations by various techniques, the reader can
find in~\cite{Rauch-Wojciechowski1996,Tsiganov1999}.

Here we present the explicit form of infinite family of admissible Lax
representations, for systems generated by~(\ref{eq4}) with $\sigma (\lambda )$
and $f(\lambda )$ being arbitrary smooth real valued functions, in
separation coordinates and in so called Vi\`{e}te coordinates.

The paper is organized as follows. In Section~\ref{sec2} we present basic facts about Benenti systems. In Section~\ref{sec3} we present in explicit form the infinite family of Lax matrices~$L(\lambda )$ and Lax equations~(\ref{eq3}) in separation coordinates and
state that the characteristic equation of each Lax matrix corresponds to the
separation curve~(\ref{eq4}) of Benenti system. All results of this section we
gather in Theorem~\ref{t1}, which we prove in Section~\ref{sec4}. Section~\ref{sec5} contains the particular Lax representations in Vi\`{e}te
coordinates in which many formulas simplify. Section~\ref{sec6} contains several examples illustrating the theory. We end the article with
some comments concerning the main result and its further application.

\section{Preliminaries}\label{sec2}
Let us consider separable systems generated by separation
curves for canonical coordinates in the form~(\ref{eq4}). The Poisson bracket
in canonical coordinates $(\lambda _{i},\mu _{i})_{i=1,\dots ,n}$ takes the canonical form
\begin{gather*}
\{f,g\}=\sum_{i=1}^{n}\left( \frac{\partial f}{\partial \lambda _{i}}\frac{\partial g}{\partial \mu _{i}}-\frac{\partial f}{\partial \mu _{i}}\frac{\partial g}{\partial \lambda _{i}}\right) .
\end{gather*}
In particular
\begin{gather*}
\{\lambda _{i},\mu _{j}\}=\delta _{ij},\qquad \{\lambda _{i},\lambda
_{j}\}=\{\mu _{i},\mu _{j}\}=0.
\end{gather*}
By taking $n$ copies of (\ref{eq4}) with $(\lambda ,\mu )=(\lambda _{i},\mu_{i})$, $i=1,\dots ,n$, we get a system of $n$ linear equations for $H_{k}$:
\begin{gather*}
H_{1}\lambda _{1}^{n-1}+H_{2}\lambda _{1}^{n-2}+\dotsb +H_{n} =F(\lambda _{1},\mu _{1}), \\
H_{1}\lambda _{2}^{n-1}+H_{2}\lambda _{2}^{n-2}+\dotsb +H_{n} =F(\lambda _{2},\mu _{2}), \\
\cdots \cdots\cdots\cdots\cdots\cdots\cdots\cdots\cdots\cdots\cdots\cdots\cdots\cdots \\
H_{1}\lambda _{n}^{n-1}+H_{2}\lambda _{n}^{n-2}+\dotsb +H_{n} =F(\lambda _{n},\mu _{n}),
\end{gather*}
where $F(x,y)=\frac{1}{2}f(x)y^{2}-\sigma (x)$. Its solution gives us $n$ real valued Hamiltonians
\begin{gather}
H_{k}(\lambda ,\mu )=E_{k}(\lambda ,\mu )+V_{k}(\lambda )=-\sum_{i=1}^{n}
\frac{\partial \rho _{k}}{\partial \lambda _{i}}\frac{F(\lambda _{i},\mu_{i})}{\Delta _{i}}, \label{eq5}
\end{gather}
and related evolution equations
\begin{gather}
\frac{{\rm d}\lambda_i}{{\rm d}t_{k}} = \{\lambda_i,H_k\}, \qquad \frac{{\rm d}\mu_i}{{\rm d}t_{k}} = \{\mu_i,H_k\}, \qquad i,k = 1,\dots,n, \label{eq6}
\end{gather}
where
\begin{gather*}
\rho _{k}=(-1)^{k}s_{k},\qquad \Delta _{i}=\prod\limits_{k\neq i}(\lambda _{i}-\lambda _{k}),
\end{gather*}%
and $s_{k}$ are elementary symmetric polynomials. From the linearity of (\ref{eq2}), geodesic parts $E_{k}$ are defined by
\begin{gather*}
\sum_{j=1}^{n}E_{j}\lambda _{i}^{n-j}=\frac{1}{2}f(\lambda _{i})\mu_{i}^{2},\qquad i=1,\dots ,n,
\end{gather*}%
and in canonical coordinates $(\lambda _{i},\mu _{i})_{i=1,\dots ,n}$ take the form \cite{Blaszak2005}
\begin{gather}
E_{j}(\lambda ,\mu )=-\frac{1}{2}\sum_{i=1}^{n}\frac{\partial \rho _{j}}{%
\partial \lambda _{i}}\frac{f(\lambda _{i})\mu _{i}^{2}}{\Delta _{i}}=\frac{1}{2}\sum_{i=1}^{n}\left( K_{j}G\right) ^{ii}\mu _{i}^{2}, \label{eq7}
\end{gather}%
where $G$ is contravariant metric tensor defined by $E_{1}$ and $K_{j}$ are
Killing tensors of $G$:
\begin{gather*}
G^{rs}=\frac{f(\lambda _{r})}{\Delta _{r}}\delta ^{rs},\qquad (K_{j})_{s}^{r}=-\frac{\partial \rho _{j}}{\partial \lambda _{r}}\delta_{s}^{r}.
\end{gather*}%
Potentials $V_{k}$, defined by
\begin{gather}
\sigma (\lambda _{i})+\sum_{j=1}^{n}V_{j}\lambda _{i}^{n-j}=0,\qquad i=1,\dots ,n \label{eq8}
\end{gather}%
take the form
\begin{gather*}
V_{k}(\lambda )=\sum_{i=1}^{n}\frac{\partial \rho _{k}}{\partial \lambda _{i}}\frac{\sigma (\lambda _{i})}{\Delta _{i}}.
\end{gather*}%
In particular, for $\sigma (\lambda )=\lambda ^{\gamma }$, $\gamma \in \mathbb{Z}$, the corresponding potentials $V_{k}^{(\gamma )}$, called basic
potentials, are constructed by the formula~\cite{Blaszak2011}
\begin{gather*}
V^{(\gamma )}=R^{\gamma }V^{(0)},\qquad V^{(\gamma )}=\big(V_{1}^{(\gamma)},\dots ,V_{n}^{(\gamma )}\big)^{\rm T},
\end{gather*}
where
\begin{gather*}
R=
\begin{pmatrix}
-\rho _{1} & 1 & 0 & 0 \\
\vdots & 0 & \ddots & 0 \\
\vdots & 0 & 0 & 1 \\
-\rho _{n} & 0 & 0 & 0
\end{pmatrix}%
,\qquad R^{-1}=%
\begin{pmatrix}
0 & 0 & 0 & -\frac{1}{\rho _{n}} \\
1 & 0 & 0 & \vdots \\
0 & \ddots & 0 & \vdots \\
0 & 0 & 1 & -\frac{\rho _{n-1}}{\rho _{n}}
\end{pmatrix}
\end{gather*}
and $V^{(0)}=(0,\dots ,0,-1)^{\rm T}$. Notice that for $\gamma =0,\dots ,n-1$
\begin{gather*}
V_{k}^{(\gamma )}=-\delta _{k,n-\gamma }
\end{gather*}%
that is
\begin{gather*}
V^{(0)}=(0,\dots ,0,-1)^{\rm T},\qquad \dots , \qquad V^{(n-1)}=(-1,0,\dots ,0)^{\rm T}.
\end{gather*}%
The first nontrivial positive potential is
\begin{gather*}
V^{(n)}=(\rho _{1},\dots ,\rho _{n})^{\rm T}
\end{gather*}%
and the negative one
\begin{gather*}
V^{(-1)}=\left( \frac{1}{\rho _{n}},\dots ,\frac{\rho _{n-1}}{\rho _{n}}\right) ^{\rm T}.
\end{gather*}

\section{Lax representation in separation coordinates}\label{sec3}
In what follows we will assume that the functions $f(\lambda)$ and $\sigma(\lambda)$ appearing in the separation curve~(\ref{eq4}) are smooth,
real valued and defined on an open subset $\mathcal{U} \subset \mathbb{R}$
such that for any phase space point $(\lambda_{1},\dots,\lambda_{n},\mu_{1},\dots,\mu_{n})$ each $\lambda_{i} \in \mathcal{U}$, $i=1,2,\dots,n$.
Moreover, we will require that $f(\lambda) \neq 0$ for every $\lambda \in \mathcal{U}$.
In order to describe the Lax representation let us investigate the division of a smooth function on~$\mathcal{U}$ by a~polynomial.

\begin{Lemma}\label{l2} Let $b(\lambda)$ be a smooth function on $\mathcal{U}$ and
\begin{gather}
a(\lambda )=(\lambda -\lambda _{1})(\lambda -\lambda _{2})\dotsm (\lambda-\lambda _{n}) \label{eq9}
\end{gather}%
be a polynomial of order $n$ whose roots $\lambda _{1},\lambda _{2},\dots,\lambda _{n}\in \mathcal{U}$, then the fraction $\frac{b(\lambda )}{a(\lambda )}$ can be uniquely written as
\begin{gather}
\frac{b(\lambda )}{a(\lambda )}=h(\lambda )+\frac{r(\lambda )}{a(\lambda )},\label{eq10}
\end{gather}%
where $h(\lambda )$ is a smooth function on $\mathcal{U}$ and $r(\lambda )$
is a polynomial of order less than $n$.
\end{Lemma}

\begin{proof}
We calculate that
\begin{align*}
\frac{b(\lambda )}{a(\lambda )}& =\frac{h_{1}(\lambda )}{(\lambda -\lambda
_{2})(\lambda -\lambda _{3})\dotsm (\lambda -\lambda _{n})}+\frac{%
h_{0}(\lambda _{1})}{a(\lambda )} \\
& =\frac{h_{2}(\lambda )}{(\lambda -\lambda _{3})(\lambda -\lambda
_{4})\dotsm (\lambda -\lambda _{n})}+\frac{h_{0}(\lambda _{1})+h_{1}(\lambda
_{2})(\lambda -\lambda _{1})}{a(\lambda )}=\dotsb \\
& =h_{n}(\lambda )+\frac{h_{0}(\lambda _{1})+h_{1}(\lambda _{2})(\lambda
-\lambda _{1})+\dotsb +h_{n-1}(\lambda _{n})(\lambda -\lambda _{1})\dotsm
(\lambda -\lambda _{n-1})}{a(\lambda )},
\end{align*}
where
\begin{gather*}
h_{0}(\lambda )=b(\lambda ),\qquad h_{i}(\lambda )=\frac{h_{i-1}(\lambda)-h_{i-1}(\lambda _{i})}{\lambda -\lambda _{i}},\qquad i=1,2,\dots ,n.
\end{gather*}
From the Taylor theorem we can see that each function $h_{i}(\lambda )$ is
smooth on $\mathcal{U}$. In particular, we have that
\begin{gather*}
\lim_{\lambda \rightarrow \lambda _{i}}\frac{{\rm d}^{k}}{{\rm d}\lambda ^{k}}h_{i}(\lambda )=\frac{1}{k+1}\frac{{\rm d}^{k+1}}{{\rm d}\lambda ^{k+1}}h_{i-1}(\lambda_{i}).
\end{gather*}
Putting $h(\lambda )=h_{n}(\lambda )$ and $r(\lambda )=h_{0}(\lambda
_{1})+h_{1}(\lambda _{2})(\lambda -\lambda _{1})+\dotsb +h_{n-1}(\lambda
_{n})(\lambda -\lambda _{1})\dotsm (\lambda -\lambda _{n-1})$ we get (\ref{eq10}). The uniqueness of decomposition (\ref{eq10}) follows from the fact that if
\begin{gather*}
\frac{b(\lambda )}{a(\lambda )}=h_{1}(\lambda )+\frac{r_{1}(\lambda )}{%
a(\lambda )}=h_{2}(\lambda )+\frac{r_{2}(\lambda )}{a(\lambda )},
\end{gather*}
then
\begin{gather*}
h_{1}(\lambda )-h_{2}(\lambda )=\frac{r_{2}(\lambda )-r_{1}(\lambda )}{a(\lambda )},
\end{gather*}%
where the left hand side is a smooth function on $\mathcal{U}$ and the right
hand side is a rational function with singularities at $\lambda _{1},\lambda
_{2},\dots ,\lambda _{n}\in \mathcal{U}$. Therefore, $h_{1}(\lambda
)-h_{2}(\lambda )=0$ and $r_{2}(\lambda )-r_{1}(\lambda )=0$.
\end{proof}

We will denote $h(\lambda )$ in the decomposition (\ref{eq10}) by $\big[\frac{b(\lambda )}{a(\lambda )}\big] _{+}$ and $r(\lambda )$ by $b(\lambda
)\mod a(\lambda )$. In practice both these functions can be calculated using
the recursive formulas derived in the proof of Lemma~\ref{l2}. In a~particular case when $b(\lambda )$ is a polynomial or a pure Laurent
polynomial we may equivalently use the division algorithms for the division
of polynomial by polynomial and the division of pure Laurent polynomial by
polynomial. For example if $b(\lambda )=\sum\limits_{k=0}^{m}b_{k}\lambda ^{k}$ is a~polynomial of order $m$ and $a(\lambda )$ is a polynomial of order $n$
given by (\ref{eq9}), such that $m\geq n$, $\big[ \frac{b(\lambda )}{a(\lambda )}\big] _{+}$ is a polynomial part of $\frac{b(\lambda )}{a(\lambda )}$ of order $(m-n)$ and $b(\lambda )\mod a(\lambda )$ is the
reminder of $\frac{b(\lambda )}{a(\lambda )}$, i.e.,
\begin{gather*}
b(\lambda )\mod a(\lambda )=\rem\left[ \frac{b(\lambda )}{a(\lambda )}\right],
\end{gather*}
being a polynomial of order less than $n$. In the division algorithm we
divide $b(\lambda )$ by the highest order term of $a(\lambda )$. On the
other hand if $b(\lambda )=c\big(\lambda ^{-1}\big)=\sum\limits_{k=1}^{m}c_{k}\lambda ^{-k}$
is a pure Laurent polynomial of order $m$ and $a(\lambda )$ is a polynomial
of order $n$ given by (\ref{eq9}), $\big[ \frac{c(\lambda ^{-1})}{a(\lambda )}\big] _{+}$ is a pure Laurent polynomial of order $m$ and $c\big(\lambda
^{-1}\big)\mod a(\lambda )$ is the reminder of $\frac{c(\lambda ^{-1})}{a(\lambda )}$ being again a polynomial of degree less than $n$. In the
division algorithm we divide $c\big(\lambda ^{-1}\big)$ by the lowest order term of~$a(\lambda )$. In the case of arbitrary Laurent polynomial $P\big(\lambda
,\lambda ^{-1}\big)=b(\lambda )+c\big(\lambda ^{-1}\big)$, the division by polynomial $a(\lambda )$ splits onto two parts described above. Note, that for a smooth
function $b(\lambda )$ and a polynomial $a(\lambda )$ there holds
\begin{gather*}
b(\lambda )=b(\lambda ) \mod a(\lambda )+a(\lambda )\left[ \frac{b(\lambda )}{a(\lambda )}\right] _{+}.
\end{gather*}

We will consider infinitely many \emph{non-equivalent} Lax matrices $L\in \mathfrak{sl}(2,\mathbb{R})$ parameterized by smooth functions $g(\lambda )$, everywhere non-zero on the same domain $\mathcal{U}$ as functions $f(\lambda )$ and $\sigma (\lambda )$. The Lax matrices in the canonical representation $(\lambda ,\mu )$, parameterized by $g(\lambda )$, are taken in the form
\begin{gather}
L(\lambda )=
\begin{pmatrix}
v(\lambda ) & u(\lambda ) \\
w(\lambda ) & -v(\lambda )
\end{pmatrix},\label{eq11}
\end{gather}
where
\begin{gather}
u(\lambda )=\prod\limits_{k=1}^{n}(\lambda -\lambda _{k})=\sum_{k=0}^{n}\rho_{k}\lambda ^{n-k},\qquad \rho _{0}\equiv 1 \label{eq12}
\end{gather}
and
\begin{align}
v(\lambda )& =\sum_{i=1}^{n}g(\lambda _{i})\mu _{i}\prod\limits_{k\neq i}%
\frac{\lambda -\lambda _{k}}{\lambda _{i}-\lambda _{k}}=\sum_{i=1}^{n}\frac{%
u(\lambda )}{\lambda -\lambda _{i}}\frac{g(\lambda _{i})\mu _{i}}{\Delta _{i}%
}\overset{(\ref{eq14})}{=}-\sum_{k=1}^{n}\left[ \sum_{i=1}^{n}\frac{\partial
\rho _{k}}{\partial \lambda _{i}}\frac{g(\lambda _{i})\mu _{i}}{\Delta _{i}}%
\right] \lambda ^{n-k} \notag \\
& =-\sum_{k=0}^{n-1}\left[ \sum_{i=1}^{n}\frac{\partial \rho _{n-k}}{%
\partial \lambda _{i}}\frac{g(\lambda _{i})\mu _{i}}{\Delta _{i}}\right]
\lambda ^{k}. \label{eq13}
\end{align}%
Notice that $u(\lambda _{i})=0$, $v(\lambda _{i})=g(\lambda _{i})\mu _{i}$
and
\begin{gather}
\Delta _{i}=u_{i}(\lambda ),\qquad u_{i}(\lambda ):=\frac{u(\lambda )}{%
\lambda -\lambda _{i}}=-\frac{\partial u(\lambda )}{\partial \lambda _{i}}%
=\prod\limits_{k\neq i}(\lambda -\lambda _{k})=-\sum_{k=1}^{n}\frac{\partial
\rho _{k}}{\partial \lambda _{i}}\lambda ^{n-k}. \label{eq14}
\end{gather}%
Moreover,
\begin{gather}
w(\lambda )=-2\frac{g^{2}(\lambda )}{f(\lambda )}\left[ \frac{F(\lambda
,v(\lambda )/g(\lambda ))}{u(\lambda )}\right] _{+}, \label{eq15}
\end{gather}%
where $F(x,y)=\frac{1}{2}f(x)y^{2}-\sigma (x)$. The function $w(\lambda )$
splits onto kinetic part $w_{E}(\lambda )$ and potential part $w_{V}(\lambda
)$ respectively:
\begin{gather}
w(\lambda )=w_{E}(\lambda )+w_{V}(\lambda )=-\frac{g^{2}(\lambda )}{%
f(\lambda )}\left[ \frac{f(\lambda )v^{2}(\lambda )/g^{2}(\lambda )}{%
u(\lambda )}\right] _{+}+2\frac{g^{2}(\lambda )}{f(\lambda )}\left[ \frac{%
\sigma (\lambda )}{u(\lambda )}\right] _{+}. \label{eq16}
\end{gather}

The main result we state in the following theorem.

\begin{Theorem}
\label{t1} For arbitrary everywhere non-zero smooth function $g(\lambda )$, separation curve \eqref{eq4} that generates dynamical systems \eqref{eq5}, is reconstructed as
follows
\begin{gather}
\det \left[ L(\lambda )-g(\lambda )\mu I\right] =0\quad \Longleftrightarrow \quad \sigma
(\lambda )+\sum_{k=1}^{n}H_{k}\lambda ^{n-1}=\frac{1}{2}f(\lambda )\mu ^{2}.\label{eq17}
\end{gather}

Moreover, Lax equations for systems \eqref{eq6}, parameterized by $g(\lambda )$, take the form
\begin{gather}
\frac{{\rm d}}{{\rm d}t_{k}}L(\lambda )=[U_{k}(\lambda ),L(\lambda )], \label{eq18}
\end{gather}%
where the Lax matrices $L(\lambda )$ are defined by \eqref{eq11}--\eqref{eq15} and
\begin{gather}
U_{k}(\lambda )=\left[ \frac{B_{k}(\lambda )}{u(\lambda )}\right] _{+},\qquad
B_{k}(\lambda )=\frac{1}{2}\frac{f(\lambda )}{g(\lambda )}\left[ \frac{u(\lambda )}{\lambda ^{n-k+1}}\right] _{+}L(\lambda ). \label{eq19}
\end{gather}
\end{Theorem}

The proof of the above theorem is involved so we present it in the separate section.

As an example let us find explicit formulas for the matrices $L(\lambda)$ and
$U_k(\lambda)$ written in a~coordinate independent way for a general St\"{a}ckel system from the
Benenti class in the case of two degrees of freedom, i.e., $n = 2$. A separation
curve of that system has the form
\begin{gather*}
\sigma(\lambda) + H_1 \lambda + H_2 = \frac{1}{2} f(\lambda) \mu^2,
\end{gather*}
where $f(\lambda)$ and $\sigma(\lambda)$ are smooth functions. According to
(\ref{eq5}) the Hamiltonians in separation coordinates are equal
\begin{gather*}
H_1(\lambda,\mu) = \frac{F(\lambda_1,\mu_1) - F(\lambda_2,\mu_2)}{\lambda_1 - \lambda_2} \\
\hphantom{H_1(\lambda,\mu)}{} = \frac{1}{2} \frac{f(\lambda_1)}{\lambda_1 - \lambda_2} \mu_1^2
- \frac{1}{2} \frac{f(\lambda_2)}{\lambda_1 - \lambda_2} \mu_2^2
+ \frac{\sigma(\lambda_2) - \sigma(\lambda_1)}{\lambda_1 - \lambda_2}, \\
H_2(\lambda,\mu) = \frac{\lambda_1 F(\lambda_2,\mu_2) - \lambda_2 F(\lambda_1,\mu_1)}{\lambda_1 - \lambda_2} \\
\hphantom{H_2(\lambda,\mu)}{} = \frac{1}{2} \frac{\lambda_1 f(\lambda_2)}{\lambda_1 - \lambda_2} \mu_2^2
- \frac{1}{2} \frac{\lambda_2 f(\lambda_1)}{\lambda_1 - \lambda_2} \mu_1^2
+ \frac{\lambda_2 \sigma(\lambda_1) - \lambda_1 \sigma(\lambda_2)}{\lambda_1 - \lambda_2}.
\end{gather*}
The Lax matrix $L(\lambda)$ parameterized by a general everywhere non-zero smooth function $g(\lambda)$ will be equal
\begin{gather*}
L(\lambda) = \begin{pmatrix}
v(\lambda) & u(\lambda) \\
w(\lambda) & -v(\lambda)
\end{pmatrix},
\end{gather*}
where in accordance to (\ref{eq12}), (\ref{eq13}), (\ref{eq15}) and (\ref{eq25})
\begin{gather*}
u(\lambda) = \lambda^2 + \rho_1 \lambda + \rho_2, \\
v(\lambda) = v_1 \lambda + v_2, \\
w(\lambda) = -\frac{v^2(\lambda)}{u(\lambda)} + \frac{2g^2(\lambda)\sigma(\lambda)}{f(\lambda)u(\lambda)}+ \frac{2g^2(\lambda)}{f(\lambda)u(\lambda)}(H_1 \lambda + H_2),
\end{gather*}
where in separation coordinates
\begin{gather*}
\rho_1 = -\lambda_1 - \lambda_2,\qquad \rho_2 = \lambda_1 \lambda_2,\\
v_1 = \frac{g(\lambda_1)\mu_1 - g(\lambda_2)\mu_2}{\lambda_1 - \lambda_2},\qquad
v_2 = \frac{\lambda_1 g(\lambda_2)\mu_2 - \lambda_2 g(\lambda_1)\mu_1}{\lambda_1 - \lambda_2}.
\end{gather*}
We find that
\begin{gather*}
U_k(\lambda) = \frac{1}{2u(\lambda)} \begin{pmatrix}
a_k(\lambda) & b_k(\lambda) \\
c_k(\lambda) & -a_k(\lambda)
\end{pmatrix}, \qquad k = 1,2,
\end{gather*}
where
\begin{alignat*}{3}
&a_1(\lambda) = \frac{f(\lambda)v(\lambda)}{g(\lambda)} + \{u(\lambda),H_1\}, \qquad &&
a_2(\lambda) = \frac{(\lambda + \rho_1)f(\lambda)v(\lambda)}{g(\lambda)} + \{u(\lambda),H_2\},& \\
&b_1(\lambda) = \frac{f(\lambda)u(\lambda)}{g(\lambda)}, \qquad &&
b_2(\lambda) = \frac{(\lambda + \rho_1)f(\lambda)u(\lambda)}{g(\lambda)},& \\
&c_1(\lambda) = \frac{f(\lambda)w(\lambda)}{g(\lambda)} - 2\{v(\lambda),H_1\}, \qquad &&
c_2(\lambda) = \frac{(\lambda + \rho_1)f(\lambda)w(\lambda)}{g(\lambda)} - 2\{v(\lambda),H_2\}.&
\end{alignat*}
Indeed, for a smooth function $b(\lambda)$, with the use of the recursive formulas derived in the proof of Lemma~\ref{l2}, we get
that
\begin{align}
\left[\frac{b(\lambda)}{u(\lambda)}\right]_+ & = \frac{\frac{b(\lambda) - b(\lambda_1)}{\lambda - \lambda_1}
- \frac{b(\lambda_2) - b(\lambda_1)}{\lambda_2 - \lambda_1}}{\lambda - \lambda_2} \nonumber \\
& = \frac{b(\lambda)}{u(\lambda)} - \frac{1}{u(\lambda)}\left(\frac{b(\lambda_1) - b(\lambda_2)}{\lambda_1 - \lambda_2} \lambda
+ \frac{\lambda_1 b(\lambda_2) - \lambda_2 b(\lambda_1)}{\lambda_1 - \lambda_2}\right)\label{eq20}
\end{align}
and
\begin{gather}
\left[\frac{b(\lambda)}{u(\lambda)}\right]_+ \bigg|_{\lambda = \lambda_k}
= (-1)^k \left(-\frac{b'(\lambda_k)}{\lambda_1 - \lambda_2} + \frac{b(\lambda_1) - b(\lambda_2)}{(\lambda_1 - \lambda_2)^2}\right), \qquad k = 1,2. \label{eq21}
\end{gather}
From (\ref{eq21}), by putting $b(\lambda) = F(\lambda,v(\lambda)/g(\lambda))$, we can calculate $w(\lambda_1)$ and $w(\lambda_2)$ and then
\begin{align}
\frac{f(\lambda_1)w(\lambda_1)}{2(\lambda_1 - \lambda_2)g(\lambda_1)} - \frac{f(\lambda_2)w(\lambda_2)}{2(\lambda_1 - \lambda_2)g(\lambda_2)} & =
 \frac{\partial v_1}{\partial \lambda_1} \frac{\partial H_1}{\partial \mu_1}
+ \frac{\partial v_1}{\partial \lambda_2} \frac{\partial H_1}{\partial \mu_2}
- \frac{\partial v_1}{\partial \mu_1} \frac{\partial H_1}{\partial \lambda_1}
- \frac{\partial v_1}{\partial \mu_2} \frac{\partial H_1}{\partial \lambda_2} \nonumber \\
& = \{v_1,H_1\}.
\label{eq22}
\end{align}
Using the fact that
\begin{alignat*}{5}
&\frac{\partial v_2}{\partial \lambda_1} = -\lambda_2 \frac{\partial v_1}{\partial \lambda_1}, \qquad &&
\frac{\partial v_2}{\partial \lambda_2} = -\lambda_1 \frac{\partial v_1}{\partial \lambda_2}, \qquad &&
\frac{\partial v_2}{\partial \mu_1} = -\lambda_2 \frac{\partial v_1}{\partial \mu_1}, \qquad &&
\frac{\partial v_2}{\partial \mu_2} = -\lambda_1 \frac{\partial v_1}{\partial \mu_2}, &\\
&\frac{\partial H_2}{\partial \lambda_1} = -\lambda_2 \frac{\partial H_1}{\partial \lambda_1}, \qquad &&
\frac{\partial H_2}{\partial \lambda_2} = -\lambda_1 \frac{\partial H_1}{\partial \lambda_2}, \qquad &&
\frac{\partial H_2}{\partial \mu_1} = -\lambda_2 \frac{\partial H_1}{\partial \mu_1}, \qquad &&
\frac{\partial H_2}{\partial \mu_2} = -\lambda_1 \frac{\partial H_1}{\partial \mu_2},&
\end{alignat*}
we also get
\begin{gather}
-\frac{\lambda_2 f(\lambda_1)w(\lambda_1)}{2(\lambda_1 - \lambda_2)g(\lambda_1)} + \frac{\lambda_1 f(\lambda_2)w(\lambda_2)}{2(\lambda_1 - \lambda_2)g(\lambda_2)} \label{eq23}\\
\qquad{} = - \lambda_2 \frac{\partial v_1}{\partial \lambda_1} \frac{\partial H_1}{\partial \mu_1}
- \lambda_1 \frac{\partial v_1}{\partial \lambda_2} \frac{\partial H_1}{\partial \mu_2}
+ \lambda_2 \frac{\partial v_1}{\partial \mu_1} \frac{\partial H_1}{\partial \lambda_1} + \lambda_1 \frac{\partial v_1}{\partial \mu_2} \frac{\partial H_1}{\partial \lambda_2}= \{v_2,H_1\} = \{v_1,H_2\}, \nonumber\\
\frac{\lambda_2^2 f(\lambda_1)w(\lambda_1)}{2(\lambda_1 - \lambda_2)g(\lambda_1)} - \frac{\lambda_1^2 f(\lambda_2)w(\lambda_2)}{2(\lambda_1 - \lambda_2)g(\lambda_2)} \nonumber\\
\qquad{} = \lambda_2^2 \frac{\partial v_1}{\partial \lambda_1} \frac{\partial H_1}{\partial \mu_1}
+ \lambda_1^2 \frac{\partial v_1}{\partial \lambda_2} \frac{\partial H_1}{\partial \mu_2}
- \lambda_2^2 \frac{\partial v_1}{\partial \mu_1} \frac{\partial H_1}{\partial \lambda_1} - \lambda_1^2 \frac{\partial v_1}{\partial \mu_2} \frac{\partial H_1}{\partial \lambda_2}
= \{v_2,H_2\}. \label{eq24}
\end{gather}
From (\ref{eq20}), by putting $b(\lambda) = \frac{f(\lambda)w(\lambda)}{2g(\lambda)}$ and $b(\lambda) = \frac{(\lambda + \rho_1)f(\lambda)w(\lambda)}{2g(\lambda)}$,
and using (\ref{eq22}), (\ref{eq23}), (\ref{eq24}) and
\begin{gather*}
\left[\frac{u(\lambda)}{\lambda}\right]_+ = \lambda + \rho_1, \qquad \left[\frac{u(\lambda)}{\lambda^2}\right]_+ = 1,
\end{gather*}
we get formulas for $c_1(\lambda)$ and $c_2(\lambda)$. Similarly we calculate $a_1(\lambda)$, $a_2(\lambda)$, $b_1(\lambda)$ and $b_2(\lambda)$.

Let us notice that for a given Lax representation $(L,U)$, with fixed $%
g(\lambda )$, there exist infinitely many \emph{gauge equivalent} Lax
representations $(L^{\prime },U^{\prime })$. Actually, let $\Omega $ be a $%
2\times 2$ invertible matrix, with matrix elements dependent on phase space
coordinates but independent on spectral parameter $\lambda $. Then, for
\begin{gather*}
L^{\prime }=\Omega L\Omega ^{-1},\qquad U^{\prime }=\Omega U\Omega ^{-1}+\Omega _{t}\Omega ^{-1}
\end{gather*}%
one can show that
\begin{gather*}
L_{t}=[U,L]\quad \Longleftrightarrow \quad L_{t}^{\prime }=[U^{\prime },L^{\prime }]
\end{gather*}%
and
\begin{gather*}
\det (L-g(\lambda )\mu I)=\det (L^{\prime }-g(\lambda )\mu I)=0.
\end{gather*}%
Hence, from the construction, such class of equivalent Lax representations
has the same $\lambda$-struc\-ture.

\section{Proof of Theorem~\protect\ref{t1}} \label{sec4}
First, let us prove the following lemma.

\begin{Lemma}\label{l1}The following equality holds
\begin{gather}
\sum_{k=1}^n H_k \lambda^{n-k} = F(\lambda,v(\lambda)/g(\lambda)) \mod u(\lambda)\label{eq25}
\end{gather}
for $F(x,y)=\frac{1}{2}f(x)y^{2}-\sigma (x)$ and $H_k$ defined by the linear system~\eqref{eq2}.
\end{Lemma}

\begin{proof}In the proof we will use the property that a polynomial of order less than $n$
is uniquely specified by its values at $n$ distinct points. The functions
$H_k = H_k(\lambda_1,\dots,\lambda_n,\mu_1,\dots,\mu_n)$ satisfy the equations~(\ref{eq2})
\begin{gather*}
\sum_{k=1}^n H_k(\lambda_1,\dots,\lambda_n,\mu_1,\dots,\mu_n) \lambda_i^{n-k}
= F(\lambda_i,\mu_i), \qquad i = 1,2,\dots,n.
\end{gather*}
For fixed $\lambda_1,\dots,\lambda_n,\mu_1,\dots,\mu_n$ such that $\lambda_i \neq \lambda_j$ for $i \neq j$ the expression
\begin{gather*}
\sum_{k=1}^n H_k(\lambda_1,\dots,\lambda_n,\mu_1,\dots,\mu_n) \lambda^{n-k}
\end{gather*}
is a polynomial in $\lambda$ of order $n - 1$, which takes values $F(\lambda_i,\mu_i)$ at $\lambda = \lambda_i$. On the other hand
\begin{gather*}
F(\lambda,v(\lambda)/g(\lambda)) \mod u(\lambda) = F(\lambda,v(\lambda)/g(\lambda))
 - u(\lambda) \left[\frac{F(\lambda,v(\lambda)/g(\lambda))}{u(\lambda)}\right]_+
\end{gather*}
is also a polynomial in $\lambda$ of order $n - 1$, which takes the same values
$F(\lambda_i,\mu_i)$ at $\lambda = \lambda_i$, since $u(\lambda_i) = 0$ and
$v(\lambda_i) = g(\lambda_i)\mu_i$. This proves the equality~(\ref{eq25}).
\end{proof}

Now we can pass to the proof of formula (\ref{eq17}).

\begin{proof}[Proof of (\ref{eq17})]
We calculate that
\begin{gather*}
\det [ L(\lambda) - g(\lambda)\mu I ] = \det \begin{pmatrix}
 v(\lambda) - g(\lambda)\mu & u(\lambda) \\
 w(\lambda) & -v(\lambda) - g(\lambda)\mu
\end{pmatrix} \\
\hphantom{\det [ L(\lambda) - g(\lambda)\mu I ]}{}
 = -(v(\lambda) - g(\lambda)\mu)(v(\lambda) + g(\lambda)\mu) - u(\lambda)w(\lambda) \\
\hphantom{\det [ L(\lambda) - g(\lambda)\mu I ]}{} = -v^2(\lambda) + g^2(\lambda)\mu^2 + 2\frac{g^2(\lambda)}{f(\lambda)} u(\lambda)
 \left[\frac{F(\lambda ,v(\lambda )/g(\lambda))}{u(\lambda)}\right]_+ \\
\hphantom{\det [ L(\lambda) - g(\lambda)\mu I ]}{} = -2\frac{g^2(\lambda)}{f(\lambda)} \left( \frac{1}{2}f(\lambda) v^2(\lambda)/g^2(\lambda)
 - \sigma(\lambda) - u(\lambda)
 \left[\frac{F(\lambda,v(\lambda)/g(\lambda))}{u(\lambda)}\right]_+ \right) \\
\hphantom{\det [ L(\lambda) - g(\lambda)\mu I ]=}{} - 2\frac{g^2(\lambda)}{f(\lambda)} \sigma(\lambda) + g^2(\lambda)\mu^2 \\
\hphantom{\det [ L(\lambda) - g(\lambda)\mu I ]}{} = -2\frac{g^2(\lambda)}{f(\lambda)} \bigl( F(\lambda,v(\lambda)/g(\lambda)) \mod u(\lambda) \bigr)
 + 2\frac{g^2(\lambda)}{f(\lambda)} \left( \frac{1}{2}f(\lambda) \mu^2 - \sigma(\lambda) \right) \\
\hphantom{\det [ L(\lambda) - g(\lambda)\mu I ]}{} = 2\frac{g^2(\lambda)}{f(\lambda)} \left( -\sum_{k=1}^n H_k \lambda^{n-k}
 + \frac{1}{2}f(\lambda) \mu^2 - \sigma(\lambda) \right),
\end{gather*}
where in the last equality we used Lemma~\ref{l1}. This proves (\ref{eq17}).
\end{proof}

Now we will show that the Lax equations (\ref{eq18}) hold. The proof is based
on the following lemmas.

\begin{Lemma}\label{l3}
The Poisson bracket of $u(\lambda)$ and $v(\lambda)$ is equal
\begin{subequations}
\begin{gather}
\{u(\lambda),u(\lambda')\} = 0, \qquad \{v(\lambda),v(\lambda')\} = 0, \label{eq26a} \\
\{u(\lambda),v(\lambda')\} = \{u(\lambda'),v(\lambda)\}= -\sum_{k=1}^n \left( g(\lambda) \left[\frac{u(\lambda)}{\lambda^{n-k+1}} \right]_+ \mod u(\lambda) \right) \lambda'^{n-k}. \label{eq26b}
\end{gather}
\end{subequations}
\end{Lemma}

\begin{proof}In the proof we will use the property that a polynomial of order less than $n$
is uniquely specified by its values at $n$ distinct points. The first equality
in (\ref{eq26a}) is straightforward. The second equality follows from
\begin{gather*}
\{v(\lambda_i),v(\lambda_j)\} = \{g(\lambda_i)\mu_i,g(\lambda_j)\mu_j\} = 0\qquad \text{for} \quad i,j = 1,2,\dots,n.
\end{gather*}
For the proof of (\ref{eq26b}) note that
\begin{align*}
\{\rho_k,v(\lambda_j)\} & = \{\rho_k,g(\lambda_j)\mu_j\}
= (-1)^k \sum_{1 \leq l_1 < l_2 < \dotsb < l_k \leq n}
 \{\lambda_{l_1} \lambda_{l_2} \dotsm \lambda_{l_k},g(\lambda_j)\mu_j\} \\
& = -g(\lambda_j) \sum_{m=0}^{k-1} \lambda_j^m \rho_{k-m-1} \{\lambda_j,\mu_j\}
= -g(\lambda_j) \sum_{m=0}^{k-1} \lambda_j^m \rho_{k-m-1}
\end{align*}
is a value of the polynomial
\begin{gather*} \left(-g(\lambda) \sum_{m=0}^{k-1} \lambda^m
\rho_{k-m-1}\right) \mod u(\lambda) \qquad \text{at}\quad \lambda = \lambda_j,
\end{gather*} since
$u(\lambda_j) = 0$. Because this polynomial is of order less than $n$ we can write
\begin{gather*}
\{\rho_k,v(\lambda)\} = \left(-g(\lambda) \sum_{m=0}^{k-1} \lambda^m
 \rho_{k-m-1}\right) \mod u(\lambda)
= -g(\lambda)\left[\frac{u(\lambda)}{\lambda^{n-k+1}}\right]_+ \mod u(\lambda).
\end{gather*}
Thus
\begin{gather*}
\{u(\lambda'),v(\lambda)\} = \sum_{k=1}^n \{\rho_k,v(\lambda)\}
 \lambda'^{n-k}
= -\sum_{k=1}^n \left( g(\lambda)
 \left[\frac{u(\lambda)}{\lambda^{n-k+1}}\right]_+ \mod u(\lambda) \right)
 \lambda'^{n-k},
\end{gather*}
which proves the second equality in (\ref{eq26b}). For the proof of the first
equality in (\ref{eq26b}) we calculate that
\begin{align*}
\{u(\lambda'),v(\lambda)\} & =
 \Biggl\{ \prod_{i=1}^n (\lambda' - \lambda_i), \sum_{l=1}^n g(\lambda_l)\mu_l
 \prod_{j \neq l} \frac{\lambda - \lambda_j}{\lambda_l - \lambda_j} \Biggr\} \\
& = \sum_{l=1}^n g(\lambda_l)
 \left\{ \prod_{i=1}^n (\lambda' - \lambda_i), \mu_l \right\}
 \prod_{j \neq l} \frac{\lambda - \lambda_j}{\lambda_l - \lambda_j} \\
& = -\sum_{l=1}^n g(\lambda_l) \{\lambda_l,\mu_l\} \prod_{j \neq l}
 \frac{(\lambda' - \lambda_j)(\lambda - \lambda_j)}{\lambda_l - \lambda_j}
= \{u(\lambda),v(\lambda')\}.\tag*{\qed}
\end{align*}\renewcommand{\qed}{}
\end{proof}

\begin{Lemma}
\label{l4}
Let $q(\lambda)$ be a polynomial of order $n$ and $p(\lambda)$ a smooth
function defined on a~domain containing all roots of $q(\lambda)$, then
\begin{gather*}
\sum_{k=1}^n \lambda'^{n-k} p(\lambda) \big[\lambda^{-n+k-1} q(\lambda)
 \big]_+ \mod q(\lambda) = \sum_{k=1}^n \lambda^{n-k}
p(\lambda') \big[\lambda'^{-n+k-1} q(\lambda')\big]_+ \mod q(\lambda').
\end{gather*}
\end{Lemma}

\begin{proof}
We have that
\begin{gather*}
p(\lambda) \big[\lambda^{-n+k-1} q(\lambda)\big]_+ \mod q(\lambda) =
 r(\lambda) \big[\lambda^{-n+k-1} q(\lambda)\big]_+ \mod q(\lambda),
\end{gather*}
where $r(\lambda) = p(\lambda) \mod q(\lambda)$. Since
$\big[\frac{r(\lambda)}{q(\lambda)}\big]_+ = 0$ it further follows that
\begin{align*}
p(\lambda) \big[\lambda^{-n+k-1} q(\lambda)\big]_+ \mod q(\lambda) & =
 r(\lambda) \big[\lambda^{-n+k-1} q(\lambda)\big]_+
 - q(\lambda) \left[\frac{r(\lambda)}{q(\lambda)}
 \left[\lambda^{-n+k-1} q(\lambda)\right]_+ \right]_+ \\
& = r(\lambda) \big[\lambda^{-n+k-1} q(\lambda)\big]_+
 - q(\lambda) \left[\frac{r(\lambda)}{q(\lambda)} \lambda^{-n+k-1} q(\lambda)
 \right]_+ \\
& = r(\lambda) \big[\lambda^{-n+k-1} q(\lambda)\big]_+
 - q(\lambda) \big[ \lambda^{-n+k-1} r(\lambda) \big]_+.
\end{align*}
Using this equality we get
\begin{gather*}
\sum_{k=1}^n \lambda'^{n-k} p(\lambda)
 \big[\lambda^{-n+k-1} q(\lambda)\big]_+ \mod q(\lambda) \\
 \qquad{} = \sum_{k=1}^n \lambda'^{n-k} \big( r(\lambda)
 \big[\lambda^{-n+k-1} q(\lambda)\big]_+
 - q(\lambda) \big[ \lambda^{-n+k-1} r(\lambda) \big]_+ \big) \\
\qquad{} = \sum_{k=1}^n \lambda^{n-k} \big( r(\lambda')
 \big[\lambda'^{-n+k-1} q(\lambda')\big]_+
 - q(\lambda') \big[ \lambda'^{-n+k-1} r(\lambda')\big]_+ \big) \\
 \qquad{} = \sum_{k=1}^n \lambda^{n-k} p(\lambda')
 \big[\lambda'^{-n+k-1} q(\lambda')\big]_+ \mod q(\lambda'),
\end{gather*}
where the second equality is easily proven by expanding the polynomials.
\end{proof}

\begin{Lemma}\label{l5}The action of Hamiltonian vector fields $X_k = \{\,\cdot\,,H_k\}$ on $u(\lambda)$ and $v(\lambda)$ is equal
\begin{subequations}\label{eq27}
\begin{gather}
X_k u(\lambda) = -\frac{\partial F}{\partial y}(\lambda,v(\lambda)/g(\lambda))
 \left[\frac{u(\lambda)}{\lambda^{n-k+1}}\right]_+ \mod u(\lambda), \label{eq27a} \\
X_k v(\lambda) = -g(\lambda) \left[\frac{F(\lambda,v(\lambda)/g(\lambda))}{u(\lambda)}\right]_+
 \left[\frac{u(\lambda)}{\lambda^{n-k+1}}\right]_+ \mod u(\lambda), \label{eq27b}
\end{gather}
\end{subequations}
where $F(x,y) = \frac{1}{2} f(x) y^2 - \sigma(x)$.
\end{Lemma}

\begin{proof}
Using Lemmas~\ref{l1}, \ref{l3} and \ref{l4} we calculate that
\begin{align*}
\sum_{k=1}^n \{u(\lambda'),H_k\} \lambda^{n-k} & =
 \{u(\lambda'), F(\lambda,v(\lambda)/g(\lambda)) \mod u(\lambda)\} \\
& = \sum_{i=1}^n \frac{\partial u(\lambda')}{\partial \lambda_i}
 \frac{\partial F}{\partial y}(\lambda,v(\lambda)/g(\lambda)) \frac{1}{g(\lambda)}
 \frac{\partial v(\lambda)}{\partial \mu_i} \mod u(\lambda) \\
& = \{u(\lambda'),v(\lambda)\} \frac{1}{g(\lambda)}
 \frac{\partial F}{\partial y}(\lambda,v(\lambda)/g(\lambda)) \mod u(\lambda) \\
& = -\sum_{k=1}^n \left(
 \frac{\partial F}{\partial y}(\lambda,v(\lambda)/g(\lambda))
 \left[\frac{u(\lambda)}{\lambda^{n-k+1}}\right]_+ \mod u(\lambda) \right)
 \lambda'^{n-k} \\
& = -\sum_{k=1}^n \left(
 \frac{\partial F}{\partial y}(\lambda',v(\lambda')/g(\lambda'))
 \left[\frac{u(\lambda')}{\lambda'^{n-k+1}}\right]_+ \mod u(\lambda') \right) \lambda^{n-k}.
\end{align*}
By comparing the coefficients of $\lambda^{n-k}$ on the left and right hand side of the above equality we get equation (\ref{eq27a}). For the proof of~(\ref{eq27b}) we first calculate
\begin{gather*}
\frac{\partial}{\partial \lambda_i}
 \bigl( F(\lambda,v(\lambda)/g(\lambda)) \mod u(\lambda) \bigr) =
 \frac{\partial}{\partial \lambda_i} \left(F(\lambda,v(\lambda)/g(\lambda)) - u(\lambda)
 \left[\frac{F(\lambda,v(\lambda)/g(\lambda))}{u(\lambda)}\right]_+ \right) \\
\qquad{} = \frac{1}{g(\lambda)} \frac{\partial F}{\partial y}(\lambda,v(\lambda)/g(\lambda))
 \frac{\partial v(\lambda)}{\partial \lambda_i}
 - \frac{\partial u(\lambda)}{\partial \lambda_i}
 \left[\frac{F(\lambda,v(\lambda)/g(\lambda))}{u(\lambda)}\right]_+ \\
\qquad\quad {} - u(\lambda) \left[\frac{1}{g(\lambda)u(\lambda)}
 \frac{\partial F}{\partial y}(\lambda,v(\lambda)/g(\lambda))
 \frac{\partial v(\lambda)}{\partial \lambda_i}
 - \frac{F(\lambda,v(\lambda)/g(\lambda))}{u^2(\lambda)}
 \frac{\partial u(\lambda)}{\partial \lambda_i}\right]_+ \\
\qquad{} = \frac{1}{g(\lambda)} \frac{\partial F}{\partial y}(\lambda,v(\lambda)/g(\lambda))
 \frac{\partial v(\lambda)}{\partial \lambda_i} \mod u(\lambda)
 - \frac{\partial u(\lambda)}{\partial \lambda_i}
 \left[\frac{F(\lambda,v(\lambda)/g(\lambda))}{u(\lambda)}\right]_+ \\
\qquad \quad {} + u(\lambda) \left[\frac{F(\lambda,v(\lambda)/g(\lambda))}{u^2(\lambda)}
 \frac{\partial u(\lambda)}{\partial \lambda_i}\right]_+.
\end{gather*}
Since $\big[\frac{1}{u(\lambda)} \frac{\partial u(\lambda)}{\partial \lambda_i} \big]_+ = 0$ we can write
\begin{align*}
\left[\frac{1}{u(\lambda)} \frac{\partial u(\lambda)}{\partial \lambda_i}
 \frac{F(\lambda,v(\lambda)/g(\lambda))}{u(\lambda)}\right]_+ & =
 \left[\frac{1}{u(\lambda)} \frac{\partial u(\lambda)}{\partial \lambda_i}
 \left[\frac{F(\lambda,v(\lambda)/g(\lambda))}{u(\lambda)}\right]_+\right]_+ \\
& = \frac{\partial u(\lambda)}{\partial \lambda_i} \left[\frac{1}{u(\lambda)}
 \left[\frac{F(\lambda,v(\lambda)/g(\lambda))}{u(\lambda)}\right]_+\right]_+
\end{align*}
and we get
\begin{align*}
\frac{\partial}{\partial \lambda_i}
 \bigl( F(\lambda,v(\lambda)/g(\lambda)) \mod u(\lambda) \bigr) ={} &
 \frac{1}{g(\lambda)} \frac{\partial v(\lambda)}{\partial \lambda_i}
 \frac{\partial F}{\partial y}(\lambda,v(\lambda)/g(\lambda)) \mod u(\lambda) \\
 & {}- \frac{\partial u(\lambda)}{\partial \lambda_i}
 \left[\frac{F(\lambda,v(\lambda)/g(\lambda))}{u(\lambda)}\right]_+ \mod u(\lambda).
\end{align*}
Using this equality and Lemmas~\ref{l1}, \ref{l3} and \ref{l4} we have
\begin{gather*}
\sum_{k=1}^n \{v(\lambda'),H_k\} \lambda^{n-k} =
 \{v(\lambda'), F(\lambda,v(\lambda)/g(\lambda)) \mod u(\lambda)\} \\
\hphantom{\sum_{k=1}^n \{v(\lambda'),H_k\} \lambda^{n-k}}{} = \sum_{i=1}^n \Biggl(
 \frac{\partial v(\lambda')}{\partial \lambda_i}
 \frac{\partial v(\lambda)}{\partial \mu_i} \frac{1}{g(\lambda)}
 \frac{\partial F}{\partial y}(\lambda,v(\lambda)/g(\lambda)) \mod u(\lambda) \\
\hphantom{\sum_{k=1}^n \{v(\lambda'),H_k\} \lambda^{n-k}=}{} - \frac{\partial v(\lambda')}{\partial \mu_i}
 \frac{\partial v(\lambda)}{\partial \lambda_i} \frac{1}{g(\lambda)}
 \frac{\partial F}{\partial y}(\lambda,v(\lambda)/g(\lambda)) \mod u(\lambda) \\
\hphantom{\sum_{k=1}^n \{v(\lambda'),H_k\} \lambda^{n-k}=}{} + \frac{\partial v(\lambda')}{\partial \mu_i}
 \frac{\partial u(\lambda)}{\partial \lambda_i}
 \left[\frac{F(\lambda,v(\lambda)/g(\lambda))}{u(\lambda)}\right]_+ \mod u(\lambda)
 \Biggr) \\
\hphantom{\sum_{k=1}^n \{v(\lambda'),H_k\} \lambda^{n-k}}{} = \{v(\lambda'),v(\lambda)\} \frac{1}{g(\lambda)}
 \frac{\partial F}{\partial y}(\lambda,v(\lambda)/g(\lambda)) \mod u(\lambda) \\
\hphantom{\sum_{k=1}^n \{v(\lambda'),H_k\} \lambda^{n-k}=}{} + \{u(\lambda),v(\lambda')\}
 \left[\frac{F(\lambda,v(\lambda)/g(\lambda))}{u(\lambda)}\right]_+ \mod u(\lambda) \\
\hphantom{\sum_{k=1}^n \{v(\lambda'),H_k\} \lambda^{n-k}}{} = -\sum_{k=1}^n \left( g(\lambda)
 \left[\frac{F(\lambda,v(\lambda)/g(\lambda))}{u(\lambda)}\right]_+
 \left[\frac{u(\lambda)}{\lambda^{n-k+1}}\right]_+ \mod u(\lambda) \right)
 \lambda'^{n-k} \\
\hphantom{\sum_{k=1}^n \{v(\lambda'),H_k\} \lambda^{n-k}}{} = -\sum_{k=1}^n \left( g(\lambda')
 \left[\frac{F(\lambda',v(\lambda')/g(\lambda'))}{u(\lambda')}\right]_+
 \left[\frac{u(\lambda')}{\lambda'^{n-k+1}}\right]_+ \mod u(\lambda') \right)
 \lambda^{n-k}.
\end{gather*}
This proves equation (\ref{eq27b}).
\end{proof}

\begin{proof}[Proof of (\ref{eq18})]
The equations (\ref{eq27}) can be rewritten in the form
\begin{subequations}\label{eq28}
\begin{gather}
X_k u(\lambda) = -\frac{f(\lambda)}{g(\lambda)} v(\lambda)
 \left[\frac{u(\lambda)}{\lambda^{n-k+1}}\right]_+
 + u(\lambda) \left[ \frac{f(\lambda)}{g(\lambda)} \frac{v(\lambda)}{u(\lambda)}
 \left[\frac{u(\lambda)}{\lambda^{n-k+1}}\right]_+ \right]_+, \\
X_k v(\lambda) = \frac{1}{2}\frac{f(\lambda)}{g(\lambda)} w(\lambda)
 \left[\frac{u(\lambda)}{\lambda^{n-k+1}}\right]_+
 - \frac{1}{2} u(\lambda) \left[ \frac{f(\lambda)}{g(\lambda)} \frac{w(\lambda)}{u(\lambda)}
 \left[\frac{u(\lambda)}{\lambda^{n-k+1}}\right]_+ \right]_+.
\end{gather}
\end{subequations}
We can now compute $X_k w(\lambda)$
\begin{align*}
X_k w(\lambda) & = -2\frac{g^2(\lambda)}{f(\lambda)} X_k
 \left[\frac{F(\lambda,v(\lambda)/g(\lambda))}{u(\lambda)}\right]_+ \\
& = -2\frac{g^2(\lambda)}{f(\lambda)} \left[ \frac{f(\lambda)}{g^2(\lambda)} \frac{v(\lambda)}{u(\lambda)}
 X_k v(\lambda) \right]_+
 + 2\frac{g^2(\lambda)}{f(\lambda)} \left[ \frac{F(\lambda,v(\lambda)/g(\lambda))}{u(\lambda)}
 \frac{X_k u(\lambda)}{u(\lambda)} \right]_+.
\end{align*}
Since $\big[\frac{X_k u(\lambda)}{u(\lambda)}\big]_+ = 0$ we can write
\begin{align}
\left[ \frac{F(\lambda,v(\lambda)/g(\lambda))}{u(\lambda)}
 \frac{X_k u(\lambda)}{u(\lambda)} \right]_+ & =
 \left[ \left[\frac{F(\lambda,v(\lambda)/g(\lambda))}{u(\lambda)}\right]_+
 \frac{X_k u(\lambda)}{u(\lambda)} \right]_+ \nonumber \\
& = -\frac{1}{2}\left[ \frac{f(\lambda)}{g^2(\lambda)} \frac{w(\lambda)}{u(\lambda)} X_k u(\lambda) \right]_+.\label{eq29}
\end{align}
By (\ref{eq28}) and (\ref{eq29}) we have
\begin{align}
X_k w(\lambda) ={} &-2\frac{g^2(\lambda)}{f(\lambda)} \left[\frac{f(\lambda)}{g^2(\lambda)}\frac{v(\lambda)}{u(\lambda)}
 X_k v(\lambda)\right]_+
 - \frac{g^2(\lambda)}{f(\lambda)} \left[ \frac{f(\lambda)}{g^2(\lambda)} \frac{w(\lambda)}{u(\lambda)}
 X_k u(\lambda) \right]_+ \nonumber \\
={}& \frac{g^2(\lambda)}{f(\lambda)} \left[ \frac{f(\lambda)}{g^2(\lambda)} v(\lambda) \left[ \frac{f(\lambda)}{g(\lambda)}
 \frac{w(\lambda)}{u(\lambda)}
 \left[\frac{u(\lambda)}{\lambda^{n-k+1}}\right]_+ \right]_+ \right]_+ \nonumber \\
& {} - \frac{g^2(\lambda)}{f(\lambda)} \left[ \frac{f(\lambda)}{g^2(\lambda)} w(\lambda) \left[ \frac{f(\lambda)}{g(\lambda)}
 \frac{v(\lambda)}{u(\lambda)}
 \left[\frac{u(\lambda)}{\lambda^{n-k+1}}\right]_+ \right]_+ \right]_+
 \nonumber \\
={} & -w(\lambda) \left[ \frac{f(\lambda)}{g(\lambda)} \frac{v(\lambda)}{u(\lambda)}
 \left[\frac{u(\lambda)}{\lambda^{n-k+1}}\right]_+ \right]_+
 + v(\lambda) \left[ \frac{f(\lambda)}{g(\lambda)} \frac{w(\lambda)}{u(\lambda)}
 \left[\frac{u(\lambda)}{\lambda^{n-k+1}}\right]_+ \right]_+.\label{eq30}
\end{align}
From (\ref{eq28}) and (\ref{eq30}) we get
\begin{gather*}
\frac{{\rm d}}{{\rm d}t_k} L(\lambda) = X_k L(\lambda) = [U_k(\lambda),L(\lambda)].\tag*{\qed}
\end{gather*}\renewcommand{\qed}{}
\end{proof}

\section{Lax representation in Vi\`{e}te coordinates}\label{sec5}
Separation coordinates $(\lambda_i,\mu_i)_{i=1,\dots,n}$ are important from the point of view of integrability of considered systems, but not practical for our purpose, as for $n > 2$, matrix elements of any Lax pair $(L,U_k)$ contain in any case a complicated rational functions. So, let us express considered systems and their Lax representations in so called Vi\`{e}te coordinates
\begin{gather}
q_i = \rho_i(\lambda), \qquad p_i = -\sum_{k=1}^n \frac{\lambda_k^{n-i} \mu_k}{\Delta_k},
\qquad i=1,\dots,n.\label{eq31}
\end{gather}
The advantage of such coordinates relies on the fact that for $\sigma(\lambda)$,
$f(\lambda)$ and $g(\lambda)$ of polynomial form, all Hamiltonians and Lax
matrix elements are polynomial functions of Vi\`{e}te coordinates as well. Here,
just for the simplicity of formulas, we will consider the particular class of
systems where $f(\lambda) = \lambda^m$, $g(\lambda) = \lambda^r$ and
$\sigma(\lambda) = \lambda^\gamma$, $m,r,\gamma \in \mathbb{Z}$.

The Hamiltonians (\ref{eq5}) take the form
\begin{gather*}
H_j = \sum_{i,k} (K_j G_m)^{ik} p_j p_k + V_j^{(\gamma)}(q),
\end{gather*}
where
\begin{gather*}
(G_m)^{ik} = -\sum_{l=0}^{k-1} q_{k-l-1} V_i^{(m+l)}(q), \qquad
(K_j)_k^i = -\sum_{l=0}^{j-1} q_{l-j-1} V_i^{(n+l-k)}(q)
\end{gather*}
and basic potentials $V^{(\gamma)}$ in $q$ coordinates are generated by the
recursion matrix
\begin{gather*}
R =
\begin{pmatrix}
-q_1 & 1 & 0 & 0 \\
\vdots & 0 & \ddots & 0 \\
\vdots & 0 & 0 & 1 \\
-q_n & 0 & 0 & 0
\end{pmatrix}
, \qquad R^{-1} =
\begin{pmatrix}
0 & 0 & 0 & -\frac{1}{q_n} \\
1 & 0 & 0 & \vdots \\
0 & \ddots & 0 & \vdots \\
0 & 0 & 1 & -\frac{q_{n-1}}{q_n}
\end{pmatrix}
.
\end{gather*}

For Lax representation, we get immediately
\begin{gather}
u(\lambda;q) = \sum_{k=0}^n q_k \lambda^{n-k}, \qquad q_0 \equiv 1.
\label{eq32}
\end{gather}
From (\ref{eq8}) and (\ref{eq31}) it follows that
\begin{gather*}
\sum_{k=1}^n \frac{\lambda_k^{n+r-i} \mu_k}{\Delta_k}
= -\sum_{j=1}^n V_j^{(n+r-i)} \sum_{k=1}^n \frac{\lambda_k^{n-j} \mu_k}{\Delta_k}
= \sum_{j=1}^n V_j^{(n+r-i)} p_j.
\end{gather*}
So,
\begin{align}
v(\lambda;q,p) & = -\sum_{k=1}^n \left( \sum_{i=1}^n
 \frac{\partial \rho_k}{\partial \lambda_i}
 \frac{\lambda_i^r \mu_i}{\Delta_i} \right) \lambda^{n-k} = \sum_{k=1}^n \left[ \sum_{i=1}^n \left( \sum_{s=0}^{k-1} \rho_s
 \frac{\lambda_i^{r+k-s-1} \mu_i}{\Delta_i}\right) \right] \lambda^{n-k}
\nonumber \\
& = \sum_{k=1}^n \left[ \sum_{s=0}^{k-1} q_s \left( \sum_{j=1}^n V_j^{(r+k-s-1)}
 p_j \right) \right] \lambda^{n-k},
\label{eq33}
\end{align}
where we used the identity
\begin{gather*}
\frac{\partial \rho_k}{\partial \lambda_i} =
 -\sum_{s=0}^{k-1} \rho_s \lambda_i^{k-s-1}.
\end{gather*}
Notice that in particular for $r = 0$
\begin{gather*}
v(\lambda;q,p) = -\sum_{k=1}^n \left[ \sum_{j=0}^{k-1} q_{k-j-1} p_{n-j} \right]
 \lambda^{n-k}.
\end{gather*}

Thus, the substitutions (\ref{eq32}), (\ref{eq33}), (\ref{eq16}) and (\ref{eq19})
in $L(\lambda;q,p)$ and $U_k(\lambda;q,p)$ lead to Lax equations (\ref{eq18}),
for $f(\lambda) = \lambda^m$, $g(\lambda) = \lambda^r$, written in canonical
$(q,p)$ coordinates.

Besides, when $f(\lambda)$ is a polynomial of order less or equal $n$, the contravariant metric tensor
\begin{gather*}
G = \diag \left( \frac{f(\lambda_1)}{\Delta_1},\dots,\frac{f(\lambda_n)}{\Delta_n}\right)
\end{gather*}
defined by $E_1$ in (\ref{eq7}) is flat, so one can pass from Vi\`{e}te
coordinates to various admissible flat coordinates \cite{Marciniak2015}.

\section{Examples}\label{sec6}

\begin{Example}Our first example is a system described by a separation curve of the canonical form
\begin{gather*}
\lambda ^{5}+H_{1}\lambda ^{2}+H_{2}\lambda +H_{3}=\frac{1}{2}\mu ^{2},
\end{gather*}
i.e., $n=3$ and $f(\lambda )=1$. This is the case for which there exist flat coordinates \cite{Blaszak2007}, related to Vi\`{e}te coordinates by
\begin{gather*}
q_{1}=x_{1},\qquad q_{2}=x_{2}+\frac{1}{4}x_{1}^{2},\qquad q_{3}=x_{3}+\frac{1}{2}x_{1}x_{2}, \\
p_{1}=y_{1}-\frac{1}{2}x_{1}y_{2}+\left( \frac{1}{4}x_{1}^{2}-\frac{1}{2}x_{2}\right) y_{3},\qquad p_{2}=y_{2}-\frac{1}{2}x_{1}y_{3},\qquad p_{3}=y_{3}.
\end{gather*}
In flat coordinates Hamiltonians are
\begin{gather*}
H_{1} =\frac{1}{2}y_{2}^{2}+y_{1}y_{3}+\frac{1}{2}x_{1}^{3}-\frac{3}{2}%
x_{1}x_{2}+x_{3}, \\
H_{2} =y_{1}y_{2}+\frac{1}{2}x_{1}y_{2}^{2}-\frac{1}{2}x_{3}y_{3}^{2}+\frac{%
1}{2}x_{1}y_{1}y_{3}-\frac{1}{2}x_{2}y_{2}y_{3}+\frac{3}{16}%
x_{1}^{4}-x_{1}x_{3}-x_{2}^{2}, \\
H_{3} =\frac{1}{2}y_{1}^{2}+\frac{1}{8}x_{1}^{2}y_{2}^{2}+\frac{1}{8}%
x_{2}^{2}y_{3}^{2}+\frac{1}{2}x_{1}y_{1}y_{2}+\frac{1}{2}x_{2}y_{1}y_{3}-%
\left( \frac{1}{4}x_{1}x_{2}+x_{3}\right) y_{2}y_{3} \\
\hphantom{H_{3}=}{}+\frac{3}{4}x_{1}^{2}x_{3}+\frac{3}{8}x_{1}^{3}x_{2}-x_{2}x_{3}-%
\frac{1}{2}x_{1}x_{2}^{2}
\end{gather*}
and Lax representation for $g(\lambda )=1$ takes the form
\begin{gather*}
L=
\begin{pmatrix}
\begin{gathered} -y_{3}\lambda ^{2}-\left( y_{2}+\tfrac{1}{2}x_{1}y_{3}\right) \lambda\vspace{1mm}\\{} -y_{1}-%
\tfrac{1}{2}x_{1}y_{2}-\tfrac{1}{2}x_{2}y_{3} \end{gathered} & \begin{gathered} \lambda ^{3}+x_{1}\lambda
^{2}+\left( \tfrac{1}{4}x_{1}^{2}+x_{2}\right) \lambda\vspace{1mm}\\{} +x_{3}+\tfrac{1}{2}%
x_{1}x_{2} \end{gathered} \vspace{3mm}\\
\begin{gathered} 2\lambda ^{2}-\big(y_{3}^{2}+2x_{1}\big)\lambda\vspace{1mm}\\{} -2y_{2}y_{3}+\tfrac{3}{2}%
x_{1}^{2}-2x_{2} \end{gathered} & \begin{gathered} y_{3}\lambda ^{2}+\left( y_{2}+\tfrac{1}{2}%
x_{1}y_{3}\right) \lambda\vspace{1mm}\\{} +y_{1}+\tfrac{1}{2}x_{1}y_{2}+\tfrac{1}{2}x_{2}y_{3} \end{gathered}%
\end{pmatrix},\\
U_{1}=
\begin{pmatrix}
0 & \frac{1}{2} \vspace{1mm}\\
0 & 0%
\end{pmatrix},\quad U_{2}=\begin{pmatrix}
-\frac{1}{2}y_{3} & \frac{1}{2}\lambda +\frac{1}{2}x_{1} \vspace{1mm}\\
1 & \frac{1}{2}y_{3}%
\end{pmatrix},\\
U_{3}=%
\begin{pmatrix}
-\frac{1}{2}y_{3}\lambda -\frac{1}{2}y_{2}-\frac{1}{4}x_{1}y_{3} & \frac{1}{2}\lambda ^{2}+\frac{1}{2}x_{1}\lambda +\frac{1}{8}x_{1}^{2}+\frac{1}{2}x_{2}
\vspace{1mm}\\
\lambda -\frac{1}{2}y_{3}^{2}-x_{1} & \frac{1}{2}y_{3}\lambda +\frac{1}{2}y_{2}+\frac{1}{4}x_{1}y_{3}%
\end{pmatrix}.
\end{gather*}
\end{Example}

\begin{Example}Our second example is a system described by a separation curve of the canonical form
\begin{gather*}
H_{1}\lambda +H_{2}=\frac{1}{2}\lambda \mu ^{2}+\lambda ^{4},
\end{gather*}%
i.e., $n=2$ and $f(\lambda )=\lambda $. This is one of the integrable cases of the H\`{e}non--Heiles system. Actually, in Cartesian coordinates, related to Vi\`{e}te coordinates by
\begin{gather*}
q_{1}=-x_{1},\qquad q_{2}=-\frac{1}{4}x_{2}^{2},\qquad p_{1}=-y_{1},\qquad p_{2}=-\frac{2y_{2}}{x_{2}},
\end{gather*}
both Hamiltonians are
\begin{gather*}
H_{1} =\frac{1}{2}y_{1}^{2}+\frac{1}{2}y_{2}^{2}+x_{1}^{3}+\frac{1}{2}%
x_{1}x_{2}^{2}, \\
H_{2} =\frac{1}{2}x_{2}y_{1}y_{2}-\frac{1}{2}x_{1}y_{2}^{2}+\frac{1}{4}%
x_{1}^{2}x_{2}^{2}+\frac{1}{16}x_{2}^{4}.
\end{gather*}
Lax representation for $g(\lambda )=1$ takes the form
\begin{gather*}
L=
\begin{pmatrix}
\frac{2y_{2}}{x_{2}}\lambda +y_{1}-\frac{2x_{1}y_{2}}{x_{2}} & \lambda
^{2}-x_{1}\lambda -\frac{1}{4}x_{2}^{2} \vspace{1mm}\\
-2\lambda -\left( \frac{4y_{2}^{2}}{x_{2}^{2}}+2x_{1}\right) +\left( \frac{%
4x_{1}y_{2}^{2}}{x_{2}^{2}}-\frac{4y_{1}y_{2}}{x_{2}}-2x_{1}^{2}-\frac{1}{2}%
x_{2}^{2}\right) \lambda ^{-1} & -\frac{2y_{2}}{x_{2}}\lambda -y_{1}+\frac{%
2x_{1}y_{2}}{x_{2}}
\end{pmatrix},\\
U_{1}=%
\begin{pmatrix}
\frac{y_{2}}{x_{2}} & \frac{1}{2}\lambda \vspace{1mm}\\
-1 & -\frac{y_{2}}{x_{2}}%
\end{pmatrix}%
,\qquad U_{2}=%
\begin{pmatrix}
\frac{y_{2}}{x_{2}}\lambda -\frac{x_{1}y_{2}}{x_{2}}+\frac{1}{2}y_{1} &
\frac{1}{2}\lambda ^{2}-\frac{1}{2}x_{1}\lambda \vspace{1mm}\\
-\lambda -\frac{2y_{2}^{2}}{x_{2}^{2}}-x_{1} & -\frac{y_{2}}{x_{2}}\lambda +%
\frac{x_{1}y_{2}}{x_{2}}-\frac{1}{2}y_{1}%
\end{pmatrix}.
\end{gather*}
Lax representation for $g(\lambda )=\lambda $ is
\begin{gather*}
L=\begin{pmatrix}
y_{1}\lambda +\frac{1}{2}x_{2}y_{2} & \lambda ^{2}-x_{1}\lambda -\frac{1}{4}x_{2}^{2} \vspace{1mm}\\
-2\lambda ^{3}-2x_{1}\lambda ^{2}-\left( 2x_{1}^{2}+\frac{1}{2}%
x_{2}^{2}\right) \lambda +y_{2}^{2} & -y_{1}\lambda -\frac{1}{2}x_{2}y_{2}
\end{pmatrix},\\
U_{1}=\begin{pmatrix}
0 & \frac{1}{2} \vspace{1mm}\\
-\lambda -2x_{1} & 0
\end{pmatrix}
,\qquad U_{2}=
\begin{pmatrix}
\frac{1}{2}y_{1} & \frac{1}{2}\lambda -\frac{1}{2}x_{1} \vspace{1mm}\\
-\lambda ^{2}-x_{1}\lambda -x_{1}^{2}-\frac{1}{2}x_{2}^{2} & -\frac{1}{2}%
y_{1}%
\end{pmatrix},
\end{gather*}
while Lax representation for $g(\lambda )=\lambda ^{2}$ is of the form
\begin{gather*}
L=
\begin{pmatrix}
\big(x_{1}y_{1}+\frac{1}{2}x_{2}y_{2}\big)\lambda +\frac{1}{4}x_{2}^{2}y_{1} &
\lambda ^{2}-x_{1}\lambda -\frac{1}{4}x_{2}^{2} \vspace{3mm}\\
\begin{gathered} -2\lambda ^{5}-2x_{1}\lambda ^{4}-\left( 2x_{1}^{2}+\tfrac{1}{2}%
x_{2}^{2}\right) \lambda ^{3}+\big(y_{1}^{2}+y_{2}^{2}\big)\lambda^{2}\vspace{1mm}\\
{}+y_{1}(x_{1}y_{1}+x_{2}y_{2})\lambda +\tfrac{1}{2}x_{1}y_{1}^{2} \end{gathered} &
-\big(x_{1}y_{1}+\frac{1}{2}x_{2}y_{2}\big)\lambda -\frac{1}{4}x_{2}^{2}y_{1}
\end{pmatrix},\\
U_{1}=
\begin{pmatrix}
-\frac{1}{2}y_{1}\lambda ^{-1} & \frac{1}{2}\lambda ^{-1} \vspace{1mm}\\
-\lambda ^{2}-2x_{1}\lambda -\big(3x_{1}^{2}+\frac{1}{2}x_{2}^{2}\big)-\frac{1}{2}%
y_{1}^{2}\lambda ^{-1} & \frac{1}{2}y_{1}\lambda ^{-1}
\end{pmatrix},\\
U_{2}=
\begin{pmatrix}
\frac{1}{2}x_{1}y_{1}\lambda ^{-1} & \frac{1}{2}-\frac{1}{2}x_{1}\lambda^{-1} \vspace{1mm}\\
\lambda ^{3}-x_{1}\lambda ^{2}-\big(x_{1}^{2}+\frac{1}{2}x_{2}^{2}\big)\lambda +
\frac{1}{2}\big(y_{1}^{2}+y_{2}^{2}-x_{1}x_{2}^{2}\big)+\frac{1}{2}
x_{1}y_{1}^{2}\lambda ^{-1} & -\frac{1}{2}x_{1}y_{1}\lambda ^{-1}
\end{pmatrix}.
\end{gather*}

Lax representations for $g(\lambda)=1$ and $g(\lambda)=\lambda^2$ are new one, at least to the knowledge of authors, while that for $g(\lambda)=\lambda $ is well known (see for example~\cite{Ravoson1993} or~\cite{Rauch-Wojciechowski1996}).
\end{Example}

\begin{Example}Our last example is a system described by a separation curve of the
canonical form
\begin{gather*}
\lambda ^{-2}+H_{1}\lambda +H_{2}=\frac{1}{2}\lambda ^{-1}\mu ^{2},
\end{gather*}%
i.e., $n=2$ and $f(\lambda )=\lambda ^{-1}$. Contrary to the previous cases
the metric defined by $H_{1}$ is non-flat. In Vi\`{e}te coordinates
\begin{gather*}
H_{1}=-\frac{1}{2}\frac{p_{1}^{2}}{q_{2}}-\frac{q_{1}p_{1}p_{2}}{q_{2}}+\frac{1}{2}\left( 1-\frac{q_{1}^{2}}{q_{2}}\right) p_{2}^{2}-\frac{q_{1}}{q_{2}^{2}},\\
H_{2}=-\frac{1}{2}\frac{q_{1}p_{1}^{2}}{q_{2}}+\left( 1-\frac{q_{1}^{2}}{q_{2}}\right) p_{1}p_{2}+\left( q_{1}-\frac{1}{2}\frac{q_{1}^{3}}{q_{2}}\right) p_{2}^{2}+\frac{1}{q_{2}}-\frac{q_{1}^{2}}{q_{2}^{2}}.
\end{gather*}
Then, the Lax representation for $g(\lambda )=1$ takes the form
\begin{gather*}
L=\begin{pmatrix}
-p_{2}\lambda -p_{1}-q_{1}p_{2} & \lambda ^{2}+q_{1}\lambda +q_{2} \vspace{1mm}\\
-\frac{(p_{1}+q_{1}p_{2})^{2}}{q_{2}}-\frac{2q_{1}}{q_{2}^{2}}+\frac{2}{q_{2}%
}\lambda ^{-1} & p_{2}\lambda +p_{1}+q_{1}p_{2}%
\end{pmatrix},
\\
U_{1}=
\begin{pmatrix}
-\frac{1}{2}\frac{p_{1}+q_{1}p_{2}}{q_{2}}\lambda ^{-1} & \frac{1}{2}+\frac{1%
}{2}q_{1}\lambda ^{-1} \vspace{1mm}\\
-\left( \frac{1}{2}\frac{(p_{1}+q_{1}p_{2})^{2}}{q_{2}^{2}}+\frac{2q_{1}}{%
q_{2}^{3}}\right) \lambda ^{-1}+\frac{1}{q_{2}^{2}}\lambda ^{-2} & \frac{1}{2%
}\frac{p_{1}+q_{1}p_{2}}{q_{2}}\lambda ^{-1}%
\end{pmatrix},
\\
U_{2}=
\begin{pmatrix}
-\frac{1}{2}\frac{q_{1}(p_{1}+q_{1}p_{2})}{q_{2}}\lambda ^{-1} & \frac{1}{2}%
\lambda ^{-1} \vspace{1mm}\\
-\left( \frac{1}{2}\frac{(p_{1}+q_{1}p_{2})^{2}-2}{q_{2}^{2}}+\frac{%
2q_{1}^{2}}{q_{2}^{3}}\right) \lambda ^{-1}+\frac{q_{1}}{q_{2}^{2}}\lambda
^{-2} & \frac{1}{2}\frac{q_{1}(p_{1}+q_{1}p_{2})}{q_{2}}\lambda ^{-1}%
\end{pmatrix},
\end{gather*}
for $g(\lambda )=\lambda ^{-1}$ we have
\begin{gather*}
L=
\begin{pmatrix}
\frac{p_{1}+q_{1}p_{2}}{q_{2}}\lambda +\frac{q_{1}(p_{1}+q_{1}p_{2})}{q_{2}}%
-p_{2} & \lambda ^{2}+q_{1}\lambda +q_{2} \vspace{3mm}\\
\begin{gathered} -\tfrac{(p_{1}+q_{1}p_{2})^{2}}{q_{2}^{2}}-\left( \tfrac{%
q_{1}(p_{1}+q_{1}p_{2})^{2}}{q_{2}^{2}}-\tfrac{2(p_{1}p_{2}+q_{1}p_{2}^{2})}{%
q_{2}}\right) \lambda ^{-1}\vspace{1mm}\\{}-\tfrac{2q_{1}}{q_{2}^{2}}\lambda ^{-2}+\tfrac{2}{%
q_{2}}\lambda ^{-3} \end{gathered} & -\frac{p_{1}+q_{1}p_{2}}{q_{2}}\lambda -\frac{%
q_{1}(p_{1}+q_{1}p_{2})}{q_{2}}+p_{2}%
\end{pmatrix},\\
U_{1}=
\begin{pmatrix}
0 & \frac{1}{2} \vspace{1mm}\\
\left( \frac{p_{1}p_{2}+q_{1}p_{2}^{2}}{q_{2}^{2}}-\frac{1}{2}\frac{%
q_{1}(p_{1}+q_{1}p_{2})^{2}+2}{q_{2}^{3}}+\frac{2q_{1}^{2}}{q_{2}^{4}}%
\right) \lambda ^{-1}-\frac{2q_{1}}{q_{2}^{3}}\lambda ^{-2}+\frac{1}{%
q_{2}^{2}}\lambda ^{-3} & 0
\end{pmatrix},
\\
U_{2}=
\begin{pmatrix}
\frac{1}{2}\frac{p_{1}+q_{1}p_{2}}{q_{2}} & \frac{1}{2}\lambda +\frac{1}{2}%
q_{1} \vspace{3mm}\\
\begin{gathered} \left( \tfrac{q_{1}p_{2}(p_{1}+q_{1}p_{2})}{q_{2}^{2}}-\tfrac{1}{2}\tfrac{%
q_{1}^{2}(p_{1}+q_{1}p_{2})^{2}+6q_{1}}{q_{2}^{3}}+\tfrac{2q_{1}^{3}}{%
q_{2}^{4}}\right) \lambda ^{-1}\vspace{1mm}\\{}+\left( \tfrac{1}{q_{2}^{2}}-\tfrac{2q_{1}^{2}}{%
q_{2}^{3}}\right) \lambda ^{-2}+\tfrac{q_{1}}{q_{2}^{2}}\lambda ^{-3} \end{gathered} & -%
\frac{1}{2}\frac{p_{1}+q_{1}p_{2}}{q_{2}}
\end{pmatrix},
\end{gather*}
and for $g(\lambda )=\lambda $
\begin{gather*}
L=
\begin{pmatrix}
-p_{1}\lambda +q_{2}p_{2} & \lambda ^{2}+q_{1}\lambda +q_{2} \vspace{1mm}\\
-\left( \frac{(p_{1}+q_{1}p_{2})^{2}}{q_{2}}-p_{2}^{2}+\frac{2q_{1}}{%
q_{2}^{2}}\right) \lambda ^{2}+\left( 2p_{1}p_{2}+q_{1}p_{2}^{2}+\frac{2}{%
q_{2}}\right) \lambda -q_{2}p_{2}^{2} & p_{1}\lambda -q_{2}p_{2}%
\end{pmatrix},
\\
U_{1}=
\begin{pmatrix}
\frac{1}{2}p_{2}\lambda ^{-2}-\frac{1}{2}\frac{p_{1}+q_{1}p_{2}}{q_{2}}%
\lambda ^{-1} & \frac{1}{2}q_{1}\lambda ^{-2}+\frac{1}{2}\lambda ^{-1} \vspace{1mm}\\
\left( \frac{p_{1}p_{2}+q_{1}p_{2}^{2}}{q_{2}}+\frac{1}{q_{2}^{2}}\right)
\lambda ^{-1}-\frac{1}{2}p_{2}^{2}\lambda ^{-2} & -\frac{1}{2}p_{2}\lambda
^{-2}+\frac{1}{2}\frac{p_{1}+q_{1}p_{2}}{q_{2}}\lambda ^{-1}%
\end{pmatrix},
\\
U_{2}=
\begin{pmatrix}
\frac{1}{2}q_{2}p_{2}\lambda ^{-2}-\frac{1}{2}\frac{%
q_{1}p_{1}-q_{2}p_{2}+q_{1}^{2}p_{2}}{q_{2}}\lambda ^{-1} & \frac{1}{2}%
\lambda ^{-2} \vspace{1mm}\\
\left( \frac{1}{2}\frac{p_{2}(2q_{1}p_{1}+2q_{1}^{2}p_{2}-q_{2}p_{2})}{q_{2}}%
+\frac{q_{1}}{q_{2}^{2}}\right) \lambda ^{-1}-\frac{1}{2}q_{1}p_{2}^{2}%
\lambda ^{-2} & -\frac{1}{2}q_{2}p_{2}\lambda ^{-2}+\frac{1}{2}\frac{%
q_{1}p_{1}-q_{2}p_{2}+q_{1}^{2}p_{2}}{q_{2}}\lambda ^{-1}%
\end{pmatrix}.
\end{gather*}
\end{Example}

\section{Final comments}\label{sec7}
The main result of the article is a systematic construction of a family of
non-equivalent Lax representations, parameterized by smooth functions
$g(\lambda)$ of spectral parameter, for arbitrary Liouville integrable system
defined by separation curve from the class (\ref{eq4}). We presented in explicit
form the Lax matrices $(L,U)$ in separation coordinates $(\lambda,\mu)$
for arbitrary $g(\lambda )$ and in so called Vi\`{e}te coordinates for
$g(\lambda) = \lambda^{r}$. It is really astonishing result that any
dynamical system (\ref{eq6}) has so large set of non-equivalent Lax
representations. It is still not clear for us what is the geometric meaning
of such freedom for the Lax pairs.

In all presented examples, our choice of $g(\lambda)$ was determined only
by the simplicity of formulas. In principle there is no obstacles to apply
more complex form of $g(\lambda)$, but then all formulas extremely complicate.
We did it for the case of two degrees of freedom in Section \ref{sec3}.

It is well known that the knowledge of Lax representation for a given
dynamical system allows to construct constants of motion and separation
coordinates. So why to study Lax pairs when we start from Liouville
integrable and separable system? We were motivated by at least two important
reasons. The first one was related with the investigation of the problem of
admissible structure of Lax pairs. It can be done in a systematic way just
in a separation coordinates and led us to the large family of non-equivalent
Lax representations. The second one is closely related with our further
research program. Actually, it relies on deformation of autonomous Liouville
integrable systems, with isospectral Lax representation, to non-autonomous
Frobenius integrable systems (Painlev\'{e} systems in particular), with
respective isomonodromic Lax representation \cite{Blaszak2019}. In order to
construct the isomonodromic Lax representation of deformed system the complete
knowledge on related isospectral Lax representation is necessary. The work is
in progress.

\subsection*{Acknowledgements}
We would like to thank the referees for many useful suggestions and constructive
criticisms concerning the first version of this article.
Z.~Doma\'nski has been partially supported by the grant 04/43/DSPB/0094 from
the Polish Ministry of Science and Higher Education.

\pdfbookmark[1]{References}{ref}
\LastPageEnding


\begin{thebibliography}{99}
\footnotesize\itemsep=0pt

\bibitem{Benenti1992}
Benenti S., Inertia tensors and {S}t\"ackel systems in the {E}uclidean
 spaces, \textit{Rend. Sem. Mat. Univ. Politec. Torino} \textbf{50} (1992),
 315--341.

\bibitem{Benenti1997}
Benenti S., Intrinsic characterization of the variable separation in the
 {H}amilton--{J}acobi equation, \href{https://doi.org/10.1063/1.532226}{\textit{J.~Math. Phys.}} \textbf{38} (1997),
 6578--6602.

\bibitem{Benenti2005}
Benenti S., Special symmetric two-tensors, equivalent dynamical systems,
 cofactor and bi-cofactor systems, \href{https://doi.org/10.1007/s10440-005-1138-9}{\textit{Acta Appl. Math.}} \textbf{87}
 (2005), 33--91.

\bibitem{Blaszak2005}
B{\l}aszak M., Separable systems with quadratic in momenta first integrals,
 \href{https://doi.org/10.1088/0305-4470/38/8/004}{\textit{J.~Phys.~A: Math. Gen.}} \textbf{38} (2005), 1667--1685,
 \href{https://arxiv.org/abs/nlin.SI/0312025}{arXiv:nlin.SI/0312025}.

\bibitem{Blaszak2019}
B{\l}aszak M., Non-autonomous {H}\'enon--{H}eiles system from {P}ainlev\'e
 class, \href{https://doi.org/10.1016/j.physleta.2019.04.025}{\textit{Phys. Lett.~A}} \textbf{383} (2019), 2149--2152,
 \href{https://arxiv.org/abs/1904.05203}{arXiv:1904.05203}.

\bibitem{Blaszak2007}
B{\l}aszak M., Sergyeyev A., Natural coordinates for a class of {B}enenti
 systems, \href{https://doi.org/10.1016/j.physleta.2007.01.001}{\textit{Phys. Lett.~A}} \textbf{365} (2007), 28--33,
 \href{https://arxiv.org/abs/nlin.SI/0604022}{arXiv:nlin.SI/0604022}.

\bibitem{Blaszak2011}
B{\l}aszak M., Sergyeyev A., Generalized {S}t\"ackel systems, \href{https://doi.org/10.1016/j.physleta.2011.05.046}{\textit{Phys.
 Lett.~A}} \textbf{375} (2011), 2617--2623.

\bibitem{Eilbeck1994}
Eilbeck J.C., Enol'skii V.Z., Kuznetsov V.B., Tsiganov A.V., Linear
 {$r$}-matrix algebra for classical separable systems, \href{ttps://doi.org/10.1088/0305-4470/27/2/038}{\textit{J.~Phys.~A:
 Math. Gen.}} \textbf{27} (1994), 567--578, \href{https://arxiv.org/abs/hep-th/9306155}{arXiv:hep-th/9306155}.

\bibitem{Marciniak2015}
Marciniak K., B{\l}aszak M., Flat coordinates of flat {S}t\"ackel systems,
 \href{https://doi.org/10.1016/j.amc.2015.06.099}{\textit{Appl. Math. Comput.}} \textbf{268} (2015), 706--716,
 \href{https://arxiv.org/abs/1406.2117}{arXiv:1406.2117}.

\bibitem{Mumford1984}
Mumford D., Tata lectures on theta.~{II}, \textit{Progress in Mathematics},
 Vol.~43, \href{https://doi.org/10.1007/978-0-8176-4578-6}{Birkh\"auser Boston, Inc.}, Boston, MA, 1984.

\bibitem{Rauch-Wojciechowski1996}
Rauch-Wojciechowski S., Marciniak K., B{\l}aszak M., Two {N}ewton
 decompositions of stationary flows of {K}d{V} and {H}arry {D}ym hierarchies,
 \href{https://doi.org/10.1016/S0378-4371(96)00220-8}{\textit{Phys.~A}} \textbf{233} (1996), 307--330.

\bibitem{Ravoson1993}
Ravoson V., Gavrilov L., Caboz R., Separability and {L}ax pairs for
 {H}\'enon--{H}eiles system, \href{https://doi.org/10.1063/1.530123}{\textit{J.~Math. Phys.}} \textbf{34} (1993),
 2385--2393.

\bibitem{Tsiganov1999}
Tsiganov A.V., Duality between integrable {S}t\"ackel systems,
 \href{https://doi.org/10.1088/0305-4470/32/45/311}{\textit{J.~Phys.~A: Math. Gen.}} \textbf{32} (1999), 7965--7982,
 \href{https://arxiv.org/abs/solv-int/9812001}{arXiv:solv-int/9812001}.

\bibitem{Vanhaecke2001}
Vanhaecke P., Integrable systems in the realm of algebraic geometry, 2nd ed.,
 \textit{Lecture Notes in Math.}, Vol.~1638, \href{https://doi.org/10.1007/3-540-44576-5}{Springer-Verlag}, Berlin,
 2001.

\end{thebibliography}
\end{document}